\documentclass[10pt, conference, letterpaper]{IEEEtran}
\IEEEoverridecommandlockouts
\usepackage[cmex10]{amsmath}
\usepackage{cite}
\usepackage{amsmath, amsthm, amssymb}
\usepackage{algorithm}
\usepackage[noend]{algpseudocode}

\usepackage{float}
\usepackage{subfig}
\usepackage{graphicx}
\usepackage{mathtools}

\usepackage{array}
\newcolumntype{C}[1]{>{\centering\arraybackslash}m{#1}}

\newtheorem{theorem}{\textbf{Theorem}}

\newtheorem{lemma}{\textbf{Lemma}}

\theoremstyle{definition}

\newtheorem{problem}{\textbf{Problem}}
\newtheorem*{problem*}{\textbf{Problem}}
\newtheorem{remark}{\textbf{Remark}}
\newtheorem{definition}{\textbf{Definition}}
\DeclareMathOperator*{\argmax}{arg\,max}
\DeclareMathOperator{\nc}{\overline{C}}
\DeclareMathOperator{\opt}{opt}
\DeclareMathOperator{\dis}{dis}
\DeclareMathOperator{\define}{\coloneqq}

\let\G\relax
\let\P\relax

\DeclareMathOperator{\E}{\mathbb{E}}

\DeclareMathOperator{\G}{\mathcal{G}}

\DeclareMathOperator{\R}{\mathcal{R}}

\DeclareMathOperator{\P}{\mathcal{P}}
\def\BibTeX{{\rm B\kern-.05em{\sc i\kern-.025em b}\kern-.08em
    T\kern-.1667em\lower.7ex\hbox{E}\kern-.125emX}}
\begin{document}

\title{Beyond Uniform Reverse Sampling: A Hybrid Sampling Technique for Misinformation Prevention \thanks{This is the online version of \cite{tong2019beyond} to fix a mistake in Sec. \ref{subsec: lower}. In particular, a new analysis of estimating the lower bound is proposed.}}

\author{\IEEEauthorblockN{
Guangmo Tong\IEEEauthorrefmark{1}
and
Ding-Zhu Du\IEEEauthorrefmark{2}}
\IEEEauthorblockA{\IEEEauthorrefmark{1}Department of Computer and Information Sciences, University of Delaware, USA\\}
\IEEEauthorblockA{\IEEEauthorrefmark{2}Department of Computer Science, University of Texas at Dallas, USA} 
\IEEEauthorblockA{Email: amotong@udel.edu, dzdu@utdallas.edu}
}

\maketitle

\begin{abstract}
Online misinformation has been considered as one of the top global risks as it may cause serious consequences such as economic damages and public panic. The misinformation prevention problem aims at generating a positive cascade with appropriate seed nodes in order to compete against the misinformation. In this paper, we study the misinformation prevention problem under the prominent independent cascade model. Due to the \#P-hardness in computing influence, the core problem is to design effective sampling methods to estimate the function value. The main contribution of this paper is a novel sampling method. Different from the classic reverse sampling technique which treats all nodes equally and samples the node uniformly, the proposed method proceeds with a hybrid sampling process which is able to attach high weights to the users who are prone to be affected by the misinformation. Consequently, the new sampling method is more powerful in generating effective samples used for computing seed nodes for the positive cascade. Based on the new hybrid sample technique, we design an algorithm offering a $(1-1/e-\epsilon)$-approximation. We experimentally evaluate the proposed method on extensive datasets and show that it outperforms the state-of-the-art solutions.
\end{abstract}


\section{Introduction}
According to the World Economic Forum\footnote{http://reports.weforum.org/}, misinformation has been one of the top global risks due to its potential for spreading fake news and malicious information. Once misinformation is detected, one feasible method is to introduce a positive cascade which is expected to reach the users before they are affected by the misinformation. The misinformation prevention (MP) problem aims at selecting effective seed nodes for the positive cascade such that the spread of misinformation can be maximally limited. The MP problem is naturally formulated as a combinatorial optimization problem and it has drawn great attention \cite{budak2011limiting,fan2013least,tong2017efficient,li2013rumor,he2012influence,ping2014sybil,tong2018misinformation}. Following this branch, we in this paper study the MP problem and our goal is to design solutions which are theoretically supported and highly effective in practice.

\textbf{Influence Maximization.} In the seminal work \cite{kempe2003maximizing} of Kempe, Kleinberg and Tardos, the well-known influence maximization (IM) problem was studied under two prominent operational models, independent cascade (IC) model and linear threshold (LT) model. The goal of the IM problem is to select seed nodes which can maximize the influence resulted by a single information cascade. There are two core techniques used for solving the IM problem. One is the study on its combinatorial properties and the other one is sampling technique utilized to speed up the algorithms. From the former, it has been shown that the IM problem is monotone and submodular, and therefore the simple greedy algorithm provides a $(1-1/e)$-approximation. Unfortunately, computing the influence of the cascade is a \#P-hard problem \cite{chen2010scalable}. As a result, the objective function cannot be efficiently computed, and we wish to estimate the function value by sampling techniques. As the most successful framework, C. Borg \textit{et al.} \cite{borgs2014maximizing} proposed the uniform reserve sampling technique which finds an unbiased estimator of the objective function via two steps: (1) uniformly select a node $v$ from the ground set; (2) simulate the diffusion process from $v$ in a reverse direction until no node can be further reached, and collect the set $S_v$ of the traversed nodes. After each sampling, we obtain a sample $S_v$ which is a subset of the users. The idea is that those in $S_v$ are the nodes that can influence $v$. Therefore, if $S_v \cap S \neq \emptyset$, it implies that $v$ can be influenced by $S$. Thus, given a collection of the samples, intuitively, the seed set $S$ which covers\footnote{We say $S$ covers a sample $S_v$ if $S \cap S_v \neq \emptyset$} the maximum number of samples should be able to maximize the objective function \cite{borgs2014maximizing}. Such a framework was later improved by \cite{tang2014influence,tang2015influence,nguyen2016stop} and it has been widely applied to other problems regarding influence diffusion in social networks.

\textbf{Misinformation Prevention.} Following the Borg's work \cite{borgs2014maximizing}, authors in \cite{tong2017efficient} designed a reverse sampling method for the MP problem. Given the seed set of the misinformation, the method proposed in \cite{tong2017efficient} proceeds with two analogous steps: (1) uniformly select a node $v$ from the ground set; (2) simulate the diffusion process from $v$ in a reverse direction in the manner of BFS until a seed of misinformation is reached, and collect the set $S_v$ of the traversed nodes. It is shown in \cite{tong2017efficient} that those in $S_v$ are in fact the nodes that can protect $v$ from being affected by the misinformation. Therefore, the nodes that can maximally cover the samples are able to best limit the spread of misinformation. 

\begin{figure}[t]
\begin{center}
\includegraphics[width=3.2in]{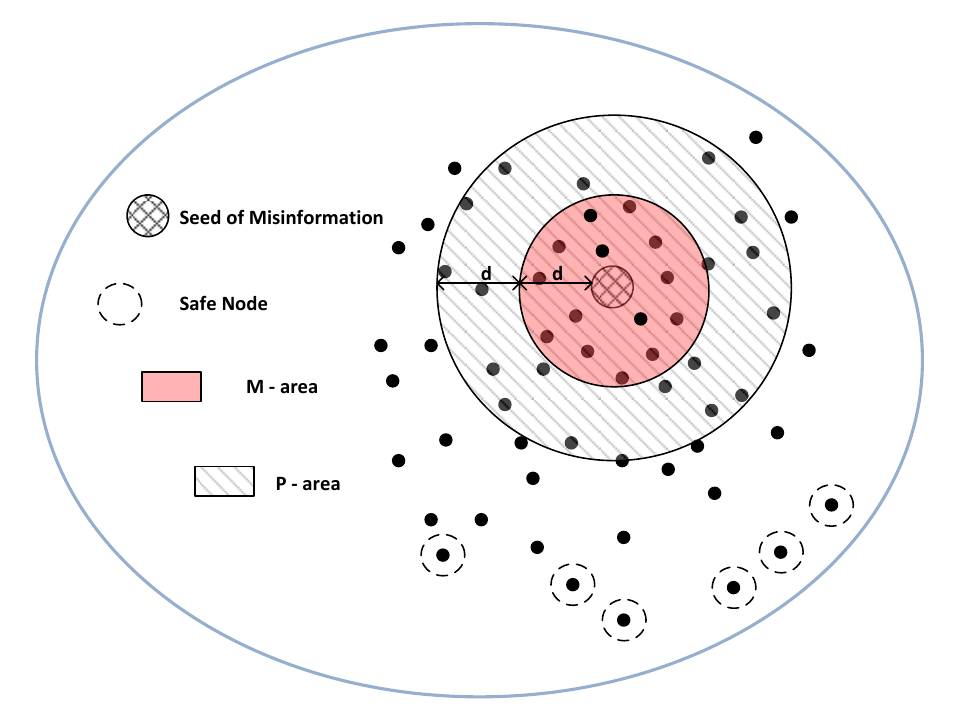} 
\end{center} 
\caption{An illustrative example}
\label{fig: example}
\vspace{-5mm}
\end{figure}

\textbf{Contribution.} Our work is driven by considering the difference between the IM problem and the MP problem. In the IM problem, all the nodes are treated equally and we cannot do better, without any prior information, than uniformly selecting a node in the first step of the sampling process. However, in the MP problem, the ultimate goal is to protect the nodes from being influenced by the misinformation, so intuitively it would be better if we could pay more attention to the nodes that are more likely to be misinformation-influenced. Consider an illustrative example shown in Fig. \ref{fig: example} where there is one seed node of the misinformation. If we uniformly select a node and do the reverse sampling, we may select the nodes that are far distant from the misinformation and they actually do not require much effort to be protected (e.g., the safe nodes in the graph). Instead, because the nodes around the seed of misinformation are the most likely to be influenced by the misinformation, we should focus more on the nodes in the M-area shown in the graph. Furthermore, because a node can be protected only if the positive cascade can reach that node before the arrival of misinformation, to protect the nodes in the M-area, the most effective positive seed nodes should lie in the P-area in the graph. Motivated by this idea, we present a hybrid sampling method consisting of two high-level steps: 
\begin{enumerate}
\item \textbf{Forward Sampling}: simulate the diffusion process of the misinformation and collect the nodes that are influenced by the misinformation.
\item \textbf{Reverse Sampling}: for each node $v$ collected in the first step, do the reverse sampling process to obtain the nodes that can protect that $v$.
\end{enumerate} 
In this sampling method, the frequency that a node can be collected in the first step is proportional to the probability that it will be affected by the misinformation. Thus, the samples produced by the second step are more likely to be the protectors of the nodes which are prone to be misinformation-influenced. Note that the pattern of the sample areas is determined by the nodes collected in the first step. As shown in Fig. \ref{fig: patterns}, under the uniform reverse sampling the samples are uniformly distributed to the whole graph, while under the hybrid sampling the samples tend to be centered around the seed node of the misinformation.  As shown later, the sample obtained by our hybrid sampling method can be used to directly estimate the prevention effect of the positive cascade. Based on the hybrid sampling method, we propose a new randomized approximation algorithm for the MP problem. In order to evaluate the proposed algorithm, we design experiments which compare the algorithms by evaluating their performance under the same time constraint. The effectiveness of our solution is supported by encouraging experimental results.

\begin{figure}[t]
\captionsetup{justification=centering}
\subfloat[Sample pattern under the uniform reverse sampling]{\label{fig: pattern_a}\includegraphics[width=0.23\textwidth]{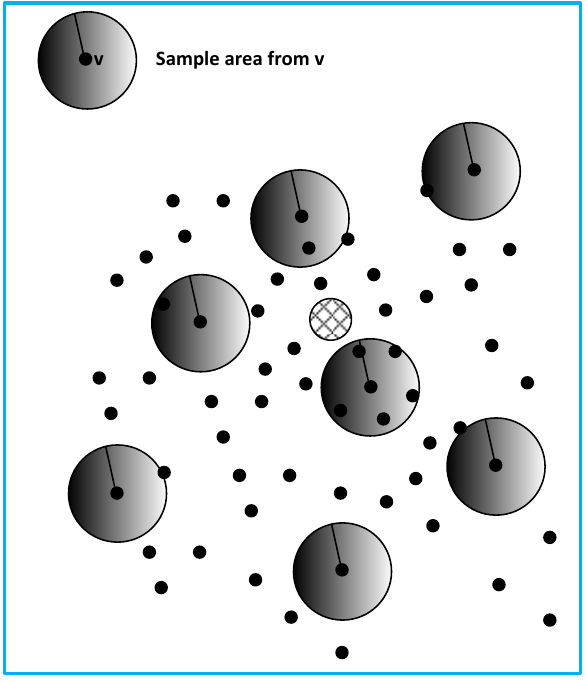}} \hspace{3mm}
\subfloat[Sample pattern under the hybrid reverse sampling]{\label{fig: pattern_b}\includegraphics[width=0.23\textwidth]{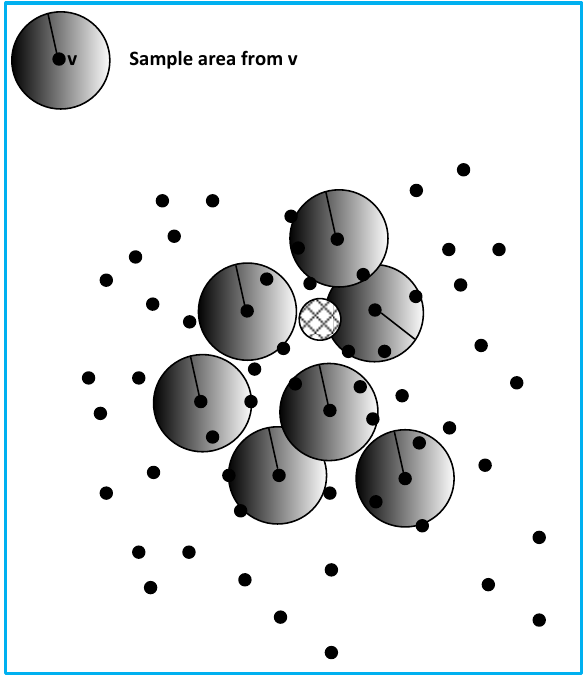}} \hspace{0mm}
\caption{Sample patterns} 
\label{fig: patterns}
\vspace{-3mm} \normalsize
\end{figure}

\textbf{Organization.} In Sec. \ref{sec: model}, we introduce the considered diffusion model and formally formulate the MP problem. We present the new hybrid sampling technique in Sec. \ref{sec: sampling}. Our algorithm is later given in Sec. \ref{sec: algortihm} in which we also provide theoretical analysis. In Sec. \ref{sec: exp}, we experimentally evaluate the proposed solution. Sec. \ref{sec: con} concludes this paper and proposes some future problems. Due to space limitation, the related works will be introduced throughout this paper and we do not spare a separate section for them.


\section{System Model}
\label{sec: model}
We consider the classic independent cascade (IC) model for influence diffusion. A social network is given by a directed graph $G=(V, E)$ where $V$ is the set of the user and $E$ denotes the relationship between users. Let $n$ and $m$ be the number of nodes and edges, respectively. Associated with each edge $(u, v) \in E$, there is a real number $\Pr[(u, v)] \in (0,1]$ denoting the propagation probability from $u$ to $v$. An information cascade $C$ starts to spread from its seed users $S \subseteq V$. We say a user is $C$-active if they are influenced by cascade $C$. All users are initialized as $\emptyset$-active, and the seed nodes of cascade $C$ are the first that are $C$-active. When a user $v$ becomes $C$-active for some cascade $C$, they have one chance to attempt to activate their $\emptyset$-active neighbor $v$ where the success probability is $\Pr[(u, v)]$. If the process succeeds,  $v$ becomes $C$-active as well. We consider the scenario where there are two cascades, the misinformation $C_r$ and the positive cascade $C_p$. We assume that the misinformation has a higher priority. That is, when a user $v$ is (a) selected as a seed node by both cascades or (b) simultaneously activated by two neighbors which are respectively $C_r$-active and $C_p$-active, $v$ will be $C_r$-active.\footnote{We note that this is not a critical setting and the analysis of this paper can be applied to the case when the positive cascade has a higher priority.} 

We assume that the seed set of the misinformation, denoted as $S_r$, is known to us, and we aim at introducing the positive cascade to spread competing against the misinformation cascade. In particular, our goal is to select an appropriate seed set for the positive cascade such that the misinformation can be prevented maximally. We say a user is $\nc_r$-active if they are not $C_r$-active. We use $f(S)$ to denote the expected number of the $\nc_r$-active nodes under a positive seed set $S \subseteq V$. The following two problems have been studied in the existing work. 

\begin{problem}[\textbf{Maximum Non-misinformation (MN) Problem}] 
\label{problem: max_non}
Given an instance of the IC model, the seed set $S_r \subseteq V$ of the misinformation, and a budget $k \in \mathbb{Z}^+$, find a subset $S \subseteq V$ with $|S| \leq k$ such that $f(S)$ is maximized.
\end{problem}

\begin{problem}[\textbf{Maximum Prevention (MP) Problem}] 
\label{problem: max_prevetion}
Given an instance of the IC model, the seed set $S_r \subseteq V$ of the misinformation, and a budget $k \in \mathbb{Z}^+$, find a subset $S \subseteq V$ with $|S| \leq k$ such that $f^*(S) \coloneqq f(S)-f(\emptyset)$ is maximized.
\end{problem}

The MN problem maximizes the expected number of the nodes which are not activated by the misinformation, while the MP problem maximizes $f(S)-f(\emptyset)$ where $f(\emptyset)$ is expected number of the $\nc_r$-active nodes when there is no positive cascade. Because $f(S)-f(\emptyset)$ denotes the exact effect of the positive cascade in preventing the misinformation, we call the second problem as the Maximum Prevention problem. These two problems are equivalent to each other with respect to the optimal solution, as their objective functions are differed by a constant $f(\emptyset)$. However, with respect to optimization, the MP problem is more challenging especially when $f(\emptyset)$ is relatively large. One can see that, for each $\alpha \in (0,1)$, each $\alpha$-approximation to the MP problem is also an $\alpha$-approximation to the MN problem. In this paper, we study the MP problem and use $S_{\opt}$ to denote its optimal solution.

\subsection{Derandomization of the diffusion process}
We say an edge $(u, v)$ is \textit{live} if $u$ can successfully activate $v$. According to the IC model, the randomness of the diffusion process comes from the states of the edges, where each edge can be live or not live with probability $\Pr[e]$ and $1-\Pr[e]$, respectively. The diffusion process is equivalent to that first sample the states of the edges and then run the diffusion process with the sampled states \cite{kempe2003maximizing}. The concept of realization is used to describe the basic event space.
\begin{definition}[\textbf{Realization}]
\label{def: realization}
Given an instance of IC model, a \textit{realization} $g$ is a two-tuple $g=(g(E_t), g(E_f))$ where $g(E_t), g(E_f) \subseteq E$ and $g(E_t) \cap g(E_f) =\emptyset$. $g(E_t)$ (resp., $g(E_f)$) denotes the set of the edges that are live (resp., not live). In addition, we say $g$ is a \textit{full-realization} if $g(E_t) \cup g(E_f) =E$.  Otherwise, it is a \textit{partial-realization}. For each realization $g$, either full or partial, we use $\Pr[g]$ to denote the probability that $g$ happens. As the states of edges are independent to each to other, we can easily check that \[\Pr[g]=\prod_{e \in g(E_t)}\Pr[e]\cdot \prod_{e \in g(E_f)}(1-\Pr[e]).\] We use $\G$ to denote the set of the full-realizations. The full-realizations provide a way to enumerate all the possible outcomes, and the partial-realizations are useful to describe an intermediate state during the diffusion process.  
\end{definition}
It is useful to interpret a full-realization $g$ as a subgraph of $G$ with the edge set $g(E_t)$. For each $g \in \G$, the diffusion process in $g$ is well-defined because the state of each edge is known to us, and furthermore, the process is deterministic. Due to the linearity of the expectation, we have $f(S)=\sum_{v \in V}f_v(S)$ where $f_v(S)$ is the probability that node $v$ is $\nc_r$-active under $S$. $f_v(S)$ can be further expressed as $\sum_{g \in \G}\Pr[g] \cdot f_v(S,g)$, where 
\[f_v(S,g) \define
  \begin{cases}
  1 &  \hspace{0mm} \hspace{-0.5mm} \text{$v$ is $\nc_r$-active in $g$ under $S$} \\
  0 & \hspace{0mm} \hspace{-0.5mm} \text{otherwise } 
  \end{cases},\]
and finally, we have 
\[f(S)=\sum_{v \in V}\sum_{g \in \G}\Pr[g] \cdot f_v(S,g).\] 
A natural problem is that under which condition that $v$ can be $\nc_r$-active in $g$. For each $V_1 \subseteq V$, $v \in V$ and $g \in \G$, we define $\dis_g(V_1, v) \define \min_{u \in V_1} \dis_g(u, v)$ where  $\dis_g(u, v)$ is the distance\footnote{\label{footnote: distance}The distant from $u$ to $v$ is measured in terms of the length of the shortest path from $u$ to $v$. We denote it as $\dis_g(u,v)=+\infty$ when there is no path from $u$ to $v$.} from $u$ to $v$ in $g$.  According to \cite{budak2011limiting}, \cite{tong2017efficient}, \cite{simpson2018scalable} and \cite{tong2018distributed}, the following lemma gives the necessary and sufficient condition for a node to be $\nc_r$-active in a full-realization.

\begin{lemma}
\label{lemma: condition_g}
Suppose the seed set of the positive cascade is $S \subseteq V$. For each $v \in V$ and $g \in \G$, $v$ is $\nc_r$-active in $g$ if and only if $\dis_g(S, v) < \dis_g(S_r,v)$ or $\dis_g(S_r,v)=+\infty$.
\end{lemma}
For each $v \in V$ and $S \subseteq V$, we use \[\G(v, S) \define \{g \in \G| f_v(S,v)=1\}\] to denote the set of the full-realization(s) $g$ where $v$ is $\nc_r$-active under $S$ in $g$.  Therefore,
\begin{equation*}
f(S)=\sum_{v \in V}\sum_{g \in \G(v, S)}\Pr[g]
\end{equation*}
and 
\begin{equation*}
f^*(S)=\sum_{v \in V}\sum_{g \in \G(v, S)}\Pr[g]-\sum_{v \in V}\sum_{g \in \G(v, \emptyset)}\Pr[g].
\end{equation*}
According to Lemma \ref{lemma: condition_g}, it is clear that $\G(v, \emptyset) \subseteq \G(v, S)$ and $\G(v, \emptyset)$ consists of the realizations where there is no path from $S_r$ to $v$. Consequently, the difference between $\G(v, \emptyset)$ and $\G(v, S)$ is the set of the realizations $g$ where $\dis_g(S, v) < \dis_g(S_r,v)$ and meanwhile there is a path from $S_r$ to $v$. We use $\G^*(v, S)$ to denote the set of such realizations, i.e., $\G^*(v, S) \define$ \[\{g\in \G| \dis_g(S, v) < \dis_g(S_r,v),\dis_g(S_r,v)\neq +\infty\} .\]
Now our objective function $f^*(S)$ can be expressed as 
\begin{equation}
\label{eq: expand_f_star}
f^*(S)=\sum_{v \in V}\sum_{g \in \G^*(v, S)}\Pr[g],
\end{equation}
which is a useful expression for the analysis later.

\section{Hybrid Sampling Technique}
\label{sec: sampling}
It has been shown in \cite{budak2011limiting} that the MP problem is monotone nondecreasing and submodular, and therefore the greedy algorithm provides a good approximation \cite{nemhauser1978analysis}. However, the greedy algorithm demands an efficient oracle of the objective function $f^*(S)$ which is \#P-hard to compute \cite{chen2010scalable}. Therefore,  it is hard to directly maximize $f^*(S)$ and, alternatively, we can first construct a good estimator of $f^*(S)$ by sampling methods, and then maximizing the obtained estimator. Towards this end, we design a novel sampling process, shown in Alg. \ref{alg: whole_sample}, which includes the following two parts: 

\textbf{Forward Sampling (line 3).} First, we ignore the positive cascade and simulate the diffusion process of the misinformation from $S_r$. When the simulation terminates, we output the immediate realization $g^*$ as well as the set $V_r$ of the nodes that are $C_r$-active. The $g^*$ produced in this step is usually a partial realization.

\textbf{Reverse Sampling (line 4-7).} Let $V_r$ and $g^*$ be the result of the first step. We consider the nodes in $V_r$ one by one in an arbitrary order and run Alg. \ref{alg: reverse_v} for each $v \in V_r$. In Alg. \ref{alg: reverse_v}, given the node $v$ and a realization $g$, we simulate the diffusion from $v$ in the reverse direction in the manner of BFS until a misinformation seed is reached. When an edge $e$ is encountered, if the state of $e$ is known in $g$ (i.e., $e \in g(E_t)\cup g(E_f)$), we keep its state; otherwise, we sample the state of $e$ and update $g$. During this process, we collect the nodes that are traversed before the misinformation seed is reached. We start with the realization $g^*$ obtained in the first step and any node in $V_r$, and repeat this process for the rest nodes in $V_r$ with the updated realization. In each run of Alg. \ref{alg: reverse_v}, a misinformation node must be reached due to the construction of $g^*$ and $V_r$. It is important to note that for each edge we sample its state for at most one time during the whole process in Alg. \ref{alg: whole_sample}. 

\begin{algorithm}[t]
\caption{Hybrid Sampling}
\label{alg: whole_sample}
\begin{algorithmic}[1]
\State \textbf{Input:} $G=(V, E)$ and $S_r$; 
\State \textbf{Output:} a collection $\P$ of subsets of $V$.
\State Simulate the diffusion of the misinformation from $S_r$. Let $V_r$ be the $C_r$-active nodes and $g^*$ be the realization when the simulation terminates;
\State $\P \leftarrow \emptyset$; 
\For {each $v \in V_r$}
\State $(P, g^*) \leftarrow$ Alg. \ref{alg: reverse_v} $(G,g^*,S_r,v)$;
\State $\P= \P \cup \{P\}$;
\EndFor
\State return $\P$; 
\end{algorithmic}
\end{algorithm}

\begin{algorithm}[t]
\caption{Reverse Sampling from $v$}
\label{alg: reverse_v}
\begin{algorithmic}[1]
\State \textbf{Input:} $G=(V, E)$, $g=(g(E_t), g(E_f))$, $S_r$ and $v$; 
\State \textbf{Output:} $P \subseteq V$ and an updated realization $g$; 
\State $P^* \leftarrow \{v\}$, $P \leftarrow \emptyset$;
\While {true}
	\If {$P^* \cap S_r \neq \emptyset$}
		\State return ($P$, $g$);
	\Else
		\State $P \leftarrow P \cup P^*$; 
		\State $E^* \leftarrow \{(u_1,u_2)| u_2 \in P^*, u_1 \in V \setminus P\}$;
		\State $P^* \leftarrow \emptyset$;
		\For {each edge $(u_1,u_2) \in E^*$}
			\If {$(u_1,u_2) \in g(E_t)$}
				\State  $P^* \leftarrow P^* \cup \{u_2\}$;
			\Else 
				\If {$(u_1,u_2) \notin g(E_f)$}
					\State $rand \leftarrow \mathcal{U}(0,1)$; \Comment{(0,1) uniform distribution} 
					\If {$rand  \leq p_{(u_1,u_2)}$}
						\State $P^* \leftarrow P^* \cup \{u_2\}$;
						\State $g(E_t) \leftarrow g(E_t) \cup \{(u_1,u_2)\}$;
					\Else
						\State $g(E_f) \leftarrow g(E_f) \cup \{(u_1,u_2)\}$;
					\EndIf 
				\EndIf 
			\EndIf 
		\EndFor 
	\EndIf 
\EndWhile 
\end{algorithmic}
\end{algorithm}

For each $V_1,V_2 \subseteq V$, we use \[y(V_1, V_2) \define
  \begin{cases}
  1 &  \hspace{0mm} \hspace{-0.5mm} \text{$V_1 \cap V_2 \neq \emptyset$} \\
  0 & \hspace{0mm} \hspace{-0.5mm} \text{otherwise } 
  \end{cases}\] 
to denote if the intersection of $V_1$ and $V_2$ is empty. Each run of Alg. \ref{alg: whole_sample} gives a family $\P$ of subsets of $V$. For each $S\subseteq V$, define $x(\P, S)$ as  $x(\P, S) \define \sum_{P \in \P} y(P, S)$.
The following lemma shows that $x(\P, S)$ is an unbiased estimator of $f^*(S)$.

\begin{lemma}
\label{lemma: key}
For each $S \subseteq V$, $\E[x(\P, x)]=f^*(S)$.
\end{lemma}
\begin{proof}
The randomness of the sampling process in Alg. \ref{alg: whole_sample} again comes from the random states of the edges. Although Alg. \ref{alg: reverse_v} are possibly called several times during Alg. \ref{alg: whole_sample}, we only update the realization $g^*$ obtained in line 3 and no edge is sampled more than once. Thus, it is equivalent to that we first sample the states of all the edges to obtain a full-realization $g$ and then run the sampling process with $g$. As a result, the basic event space is again $\G$, and therefore we can express $E[x(\P, S)]$ as $E[x(\P, S)]=\sum_{g \in \G} \Pr[g]\cdot x(\P_g, S)$ where $\P_g$ is output of Alg. \ref{alg: whole_sample} when the realization is fixed as $g$. When the realization is fixed as $g$, we use $V_r^g$ to denote the $V_r$ obtained in line 3 of Alg. \ref{alg: whole_sample}, and let $P_v^g$ be subset returned by Alg. \ref{alg: reverse_v} with an input $v$. Now we have 
\begin{eqnarray}
\label{eq: expand}
E[x(\P, S)]&=&\sum_{g \in \G} \Pr[g]\cdot x(\P_g, S) \nonumber\\
&=&\sum_{g \in \G} \Pr[g]\sum_{v \in V_r^g} y(P_v^g, S)
\end{eqnarray}
Note that $v$ is contained in $V_r^g$ if and only if $v$ is $C_r$-active in $g$ when there is no positive cascade. That is, $v \in V_r^g$ if and only if $\dis_g(C_r,v) \neq +\infty$. Thus, we have
\begin{eqnarray*}
E[x(\P, S)]&=& \sum_{g \in \G} \Pr[g]\sum_{v \in V_r^g} y(P_v^g, S)\\
&=&\sum_{v \in V}\sum_{g \in \G^*(v, \emptyset)} \Pr[g]\cdot y(P_v^g, S)
\end{eqnarray*}
Furthermore, in Alg. \ref{alg: reverse_v} we search the nodes from $v$ in a reverse direction in the manner of BFS, so $P_v^g$ contains exactly the node(s) $u$ with a distance to $v$ shorter than $\dis_g(S_r,v)$. Therefore, $y(P_v^g, S)=1$ if and only if $\dis_g(S,v)<\dis_g(S_r,v)$, which means $g \in \G^*(v,S)$. Since $\G^*(v,S) \subseteq \G^*(v,\emptyset)$, we have
\begin{eqnarray*}
E[x(\P, S)]&=&\sum_{v \in V}\sum_{g \in \G^*(v, \emptyset)} \Pr[g]\cdot y(P_v^g, S)\\
&=&\sum_{v \in V}\sum_{g \in \G^*(v, S)} \Pr[g]\cdot 1\\
&&\{\text{By Eq. (\ref{eq: expand_f_star})}\}\\
&=&f^*(S).
\end{eqnarray*} 
Thus, proved.
\end{proof} 

For convenience, we call each family $\P$ of the subsets of $V$ produced by Alg. \ref{alg: whole_sample} as an R-sample. The complexity of generating one R-sample is given as follows.

\begin{lemma}
\label{lemma: cost-R}
The expected running time of generating one R-sample is $O(m+m\cdot \max_1)$ where $\max_1 \define\max f^* (\{v\} )$.
\end{lemma}
\begin{proof}
Let $T$ be the expected running time of Alg. \ref{alg: whole_sample}, $T_{g}$ be the running time when the realization is fixed as $g$. When the realization is fixed as a full-realizatoin $g$, for each $v \in V$, we use $TIME_g$ to denote the running time of line 3, and $TIME_g^v$ to denote the running time of Alg. \ref{alg: reverse_v} with input $v$. Following the proof of Lemma \ref{lemma: key}, we have the notations $V_r^g$ and $P_v^g$.

According to the analysis  in proof of Lemma \ref{lemma: key}, we have 
\[T=\sum_{g \in \G} \Pr[g]\cdot T_{g}=\sum_{g \in \G} \Pr[g]\cdot (TIME_g + \sum_{v \in V_r^g} TIME_g^v).\]
Because $TIME_g$ is asymptotically bounded by the number of the traversed edges before the simulation terminates, we have 
\[\sum_{g \in \G} \Pr[g]\cdot TIME_g \leq \sum_{g \in \G} \Pr[g]\cdot m=O(m).\] Furthermore, according to Alg. \ref{alg: reverse_v}, an edge $(u, v)$ will be checked in Alg. \ref{alg: reverse_v} if and only if $u \in P_g^v$, i.e., $y(P_v^g, \{u\})=1$. Thus, we have \[TIME_g^v=\sum_{ (u,v)\in E}y(P_v^g, \{u\}),\] and therefore,
\[T=O(m)+\sum_{(u,v)\in E}\sum_{g \in \G} \Pr[g]\cdot \sum_{v \in V_r^g} y(P_v^g, \{u\}).\]
For each $u \in V$, by Eq. (\ref{eq: expand}), $\sum_{g \in \G} \Pr[g]\cdot \sum_{v \in V_r^g} y(P_v^g, \{u\})$ is equal to $E[x(\P, \{u\})]$ which is actually $f^*(\{u\})$. As a result, we have \[T \leq m \cdot \Big(1+\max\big(f^*(\{u\})\big)\Big).\] Thus, proved.
\end{proof}
 
\begin{algorithm}[t]
\caption{Greedy}
\label{alg: greedy}
\begin{algorithmic}[1]
\State \textbf{Input:} $\R_l$, $V$, and $k \in \mathbb{Z}^+$; 
\State $S^* \leftarrow \emptyset$;
\While {$|S^*| < k$}
\State $v^* \leftarrow \argmax_{v \in V} \overline{x}(\R_l, S^*\cup \{v\})$;
\State $S^* \leftarrow S^* \cup \{v\}$;
\EndWhile
\Return $S^*$;
\end{algorithmic}
\end{algorithm}

\section{The Algorithm}
\label{sec: algortihm}
In this section, we present the algorithm designed based on the hybrid sampling method for solving the MP problem and provide theoretical analysis.

\subsection{The Idea.} 
According to Lemma \ref{lemma: key}, $x(\P, S)$ is an unbiased estimator of $f_2(S)$. Therefore, if we generate a collection $\R_l=\{\P_1,...,\P_l\}$ of $l$ R-samples by running Alg. \ref{alg: whole_sample}, \[\overline{x}(\R_l,S)\coloneqq \frac{\sum_{\P \in \R_{l}} x(\P, S)}{l}\] should be an accurate estimate of $f^*(S)$ when $l$ is sufficiently large. As a result, the subset $S \subseteq V$ that can maximize $\overline{x}(\R_l,S)$ is intuitively a good solution to maximizing $f^*(S)$. 

\subsection{Maximizing $\overline{x}(\R_l,S)$} 
Let us first consider that how to maximize $\overline{x}(\R_l,S)$. The following result shows that $\overline{x}(\R_l,S)$ is monotone nondecreasing and submodular. 
\begin{lemma}
\label{lemma: submodular}
For each $\R_l$, $\overline{x}(\R_l, S)$ is monotone nondecreasing and submodular.
\end{lemma}
\begin{proof}
Recall that \[\overline{x}(\R_l, S)=\frac{\sum_{\P \in \R_l}x(\P,S)}{l}=\frac{\sum_{\P \in \R_l}\sum_{P \in \P}y(P,S)}{l}.\] It is monotone nondecreasing because adding a new node to $S$ does not decrease $y(P,S)$. Since the submodularity is preserved under addition, it suffices to show that $y(P,S)$ is submodular with respect to $S$. That is, for each $P \subseteq V$, $S_1 \subseteq S_2 \subseteq V$ and $v \notin S_2$, we have 
\[y(P,S_1 \cup \{v\})-y(P, S_1) \geq y(P, S_2 \cup \{v\})-y(P, S_2).\]
Because $y$ is either 0 or 1 and it is monotone, it suffices to prove that $y(P, S_1 \cup \{v\})-y(P, S_1)$ is equal 1 whenever $y(P, S_2 \cup \{v\})-y(P, S_2)$ is equal to 1. When $y(P, S_2 \cup \{v\})-y(P, S_2)$ is equal 1, we have $y(P, S_2 \cup \{v\})=1$ and $y(P, S_2)=1$. That is, $P \cap (S_2 \cup \{v\}) \neq \emptyset$ and $P \cap S_2 = \emptyset$, which implies that $v \in P$. As a result, $P \cap S_1 = \emptyset$ and $P \cap (S_1 \cup \{v\}) \neq \emptyset$. Therefore, $y(P, S_1 \cup \{v\})-y(P, S_1)$ is also equal to 1.
\end{proof}

Since $\overline{x}(\R_l, S)$ is monotone and submodular, due to the classic result in \cite{nemhauser1978analysis}, the greedy  algorithm given in Alg. \ref{alg: greedy} provides a $(1-1/e)$-approximation for maximizing $\overline{x}(\R_l, S)$.

\begin{lemma}
\label{lemma: greedy}
For each $\R_l=\{\P_1,..., \P_l\}$ and $S \subseteq V$ with $|S|\leq k$, the $S^*$ produced by Alg. \ref{alg: greedy} satisfies that $\overline{x}(\R_l, S^*) \geq (1-1/e)\cdot \overline{x}(\R_l, S)$.
\end{lemma}

\begin{algorithm}[t]
\caption{Framework}
\label{alg: framework}
\begin{algorithmic}[1]
\State \textbf{Input:} $G=(V, E)$, $S_r$, $k$,  and $l$; 
\State Generate $l$ R-samples $\R_l=\{\P_1,...,\P_l\}$ by Alg. \ref{alg: whole_sample};
\State Obtain an $S^*$ by Alg. \ref{alg: greedy} with $\R_l$, $V$ and $k$; 
\State \textbf{Return} $S^*$;
\end{algorithmic}
\end{algorithm}

\subsection{The Framework}
By taking Alg. \ref{alg: greedy} as a subroutine, the framework of our algorithm is shown in Alg. \ref{alg: framework}. Given a threshold $l \in \mathbb{Z}^+$, we first generate $l$ R-samples and then run the greedy algorithm to obtain a seed set $S^*$ of the positive cascade. Now the only part left is to determine the number of R-samples to be used. We will discuss this problem in the next subsection.

\subsection{Determining $l$}
The goal is to set $l$ sufficiently large in order to ensure the theoretical guarantees. Since our algorithm is randomized, given any parameters $\epsilon \in (0,1)$ and $N>1$, the goal is to produce an $S^*$ such that \[f^*(S^*) \geq (1-1/e-\epsilon) f^* (S_{\opt})\] holds with probability at least $1-1/N$.\footnote{As discussed in \cite{budak2011limiting}, $1-1/e$ is the best possible approximation ratio.} Throughout our analysis, we assume that $\epsilon$ and $N$ are fixed, and we aim to achieve an arbitrary small error $\epsilon$. We use the Chernoff bound \cite{motwani2010randomized} to analyze the accuracy of the estimates.

\begin{lemma}[\textbf{Chernoff Bound}]
Let $X_1, ..., X_l \in [0,1]$ be $l$ i.i.d random variables where $E(X_i)=\mu$. The Chernoff bound states that 
\begin{equation}
\label{eq:chernoff_1}
\mathrm{Pr}\Big[\sum X_i - l \cdot \mu \geq \delta \cdot l \cdot  \mu  \Big]\leq \exp(-\frac{l \cdot \mu \cdot \delta^2}{2+\delta}),
\end{equation}
and
\begin{equation}
\label{eq:chernoff_2}
\mathrm{Pr}\Big[\sum X_i - l \cdot \mu \leq -\delta \cdot l \cdot \mu \Big]\leq \exp(-\frac{l \cdot \mu \cdot \delta^2}{2}),
\end{equation}  
for each $\delta >0$.
\end{lemma}
In particular, we need $\overline{x}(\R_l, S)$ to be the accurate estimates for the solution $S^*$ produced by the greedy algorithm as well as for the optimal solution $S_{\opt}$. We use $\epsilon_{11}>0$ and $\epsilon_{12}>0$ to control the accuracy of $\overline{x}(\R_l, S^*)$ and $\overline{x}(\R_l, S_{\opt})$, respectively. The following conditions will be sufficient.

\begin{align}
\label{eq: cond_1}
&\overline{x}(\R_l, S^*)-f^*(S^*) \leq \epsilon_{11} (1+\epsilon)  f^*(S_{\opt})  \\
\label{eq: cond_2}
&\overline{x}(\R_l, S_{\opt})- f^*(S_{\opt}) \geq -\epsilon_{12} (1+\epsilon)  f^*(S_{\opt})  \\
\label{eq: cond_3}
&(1-\epsilon_{12}(1+\epsilon))(1-1/e)-\epsilon_{11}(1+\epsilon)=1-1/e-\epsilon \\
\label{eq: cond_4}
&\epsilon_{11}, \epsilon_{12} \in (0,1) 
\end{align}

\begin{lemma}
\label{lemma: condition}
Whenever Eqs. (\ref{eq: cond_1})-(\ref{eq: cond_4}) hold, we have \[f^*(S^*) \geq (1-1/e-\epsilon) f^* (S_{\opt}).\]
\end{lemma}
\begin{proof}
It can be proved by simple rearrangements. By Eq. (\ref{eq: cond_1}), we have \[f^*(S^*) \geq \overline{x}(\R_l, S^*)-\epsilon_{11} (1+\epsilon)  f^*(S_{\opt}).\] 
Due to Lemma \ref{lemma: greedy}, we have $\overline{x}(\R_l, S^*) \geq (1-1/e) \overline{x}(\R_l, S_{\opt})$ and therefore, 
\[f^*(S^*) \geq (1-1/e) \overline{x}(\R_l, S_{\opt})-\epsilon_{11} (1+\epsilon)  f^*(S_{\opt}).\] 
Finally, by Eqs. (\ref{eq: cond_2}) and (\ref{eq: cond_3}), we have $f^*(S^*) \geq (1-1/e)(1-\epsilon_{12}(1+\epsilon))  f^*(S_{\opt})- \epsilon_{11} (1+\epsilon)  f^*(S_{\opt}) \geq (1-1/e-\epsilon)  f^*(S_{\opt})$. 
\end{proof}

We take the following setting of $l$ to make Eqs. (\ref{eq: cond_1}) and (\ref{eq: cond_2}) satisfied with probability at least $1-2/N$.
\[l^* \define \frac{\max(l_1, l_2)}{(1+\epsilon)^2 f^*(S_{\opt})}\] where
\[l_1 \define \dfrac{n (\ln \binom{n}{k} + \ln N)(2+\epsilon_{11}(1+\epsilon))}{(\epsilon_{11})^{2}}\] and
\[l_2 \define \frac{2\cdot n \cdot \ln N }{(\epsilon_{12})^2}.\] 

\begin{lemma}
\label{lemma: l}
$f^*(S^*) \geq (1-1/e-\epsilon) f^* (S_{\opt})$ with probability at least $1-2/N$ provided that $l\geq l^*$.
\end{lemma}
\begin{proof}
For each $S \subseteq V$,
\begin{eqnarray*}
&&\Pr[\overline{x}(\R_l, S)-f^*(S) \geq \epsilon_{11}   (1+\epsilon)f^*(S_{\opt})] \\
&=&\Pr\Big[\frac{\overline{x}\tiny(\R_l, S\tiny)}{n}-\frac{f^*\tiny(S\tiny)}{n} \geq \frac{f^*\tiny(S\tiny)}{n}\cdot \frac{\epsilon_{11}(1+\epsilon) \cdot f^*\tiny(S_{\opt}\tiny)}{f^*\tiny(S\tiny)}\Big] \\
&&\{\text{By the Chernoff bound}\}\\
&\leq& \exp\Big(-\dfrac{l\cdot \frac{f^*(S)}{n} \cdot  \big(\frac{\epsilon_{11} \cdot (1+\epsilon) f^*(S_{\opt})}{f^*(S)}\big)^2}{2+ \frac{\epsilon_{11}  \cdot(1+\epsilon) f^*(S_{\opt})}{f^*(S)}}\Big)\\
&=& \exp\Big(-\dfrac{l\cdot \big(\epsilon_{11}  \cdot (1+\epsilon) f^*(S_{\opt})\big)^2 }{2\cdot n \cdot f^*(S)+   n \cdot \epsilon_{11}  \cdot (1+\epsilon) f^*(S_{\opt})}\Big)\\
&&\{\text{Since $f^*(S) \leq f^*(S_{\opt})$ }\}\\
&\leq& \exp\Big(-\dfrac{l \cdot \epsilon_{11}^2 (1+\epsilon)^2\cdot f^*\tiny(S_{\opt}\tiny) }{2\cdot n +   n \cdot \epsilon_{11} (1+\epsilon)}\Big)\\
&&\{\text{Since $l \geq \frac{l_1}{(1+\epsilon)^2 \cdot f^*(S_{\opt})}$ }\}\\
&\leq& \frac{1}{\binom{n}{k}\cdot N}\\
\end{eqnarray*}
Note that $|S^*|=k$ and there are at most $\binom{n}{k}$ subsets of $V$ with size $k$. By the union bound, with probability at least $1-1/N$, we have $\overline{x}(\R_l, S^*)-f^*(S^*) \leq \epsilon_{11}(1+\epsilon) \cdot f^*(S_{\opt})$. Here we have to consider all the subsets with size $k$ because the solution produced by Alg. \ref{alg: greedy} is unknown to us before generating the R-samples.

Second, applying the Chernoff bound immediately yields that when $l \geq \frac{l_2}{(1+\epsilon)^2 f^*(S_{\opt})}$, with probability at least $1-1/N$, we have $\overline{x}(\R_l, S_{\opt})- f^*(S_{\opt}) \geq -\epsilon_{12}  (1+\epsilon)  f^*(S_{\opt})$. 

Since $l \geq l^*$, by the union bound, $f^*(S^*) \geq (1-1/e-\epsilon) f^* (S_{\opt})$ with probability at least $1-2/N$.
\end{proof}

So far $\epsilon_{11}>0$ and $\epsilon_{12}>0$ are free parameters as long as they satisfy that Eqs. (\ref{eq: cond_3}) and (\ref{eq: cond_4}). Note that Eqs. (\ref{eq: cond_3}) and (\ref{eq: cond_4}) have at least one feasible solution. Ideally we should minimize $l$, so the optimal setting can be obtained by solving the following problem.
\begin{equation}
\label{eq: optimal}
\begin{aligned}
& {\text{minimize}}
& & \max(l_1,l_2) \\
& \text{subject to} & & \text{Eqs. (\ref{eq: cond_3}) and (\ref{eq: cond_4})}
\end{aligned}
\end{equation}

Now the only unknown part of $l^*$ is $f^*(S_{\opt})$. According to Lemma \ref{lemma: l}, for each $OPT_L$ with $OPT_L \leq (1+\epsilon)^2   f^*(S_{\opt})$, setting $l$ as $\max(l_1, l_2)/OPT_L$ will be sufficient. In the next subsection, we will show how to obtain such a $OPT_L$.

\subsection{Lower Bound of $f^*(S_{opt})$}
\label{subsec: lower}

Note that for each $S_L \subseteq V$ with $|S_L|=k$, $f^*(S_L)$ is a lower bound of $f^*(S_{\opt})$, and according to Lemma \ref{lemma: key}, $f^*(S_{L})$ can be estimated by the hybrid sampling process. The following result establishes the relationship between the number of samples and the accuracy.

\begin{lemma}[Dagum \textit{et al.} \cite{dagum2000optimal}]
\label{lemma: optimal_monte}
For each $0 \leq \epsilon_0 \leq 1$ and $N \geq 0$, there exists an algorithm which, with probability at least $1-1/N$, produces a $f^*_{\epsilon_0}(S_{L})$ such that  
\begin{equation}
\label{eq: opt_lower}
(1-\epsilon_0)   f^*(S_{L}) \leq  f^*_{\epsilon_0}(S_{L}) \leq (1+\epsilon_0)  f^*(S_{L})
\end{equation}
where the number of the used simulations is asymptotically bounded by $O(\frac{\epsilon^2_0+4(e-2)(1+\epsilon_0)\ln(N/2)}{\epsilon^2_0 \cdot \frac{f^*(S_{L})}{n}})$. 
\end{lemma}

Taking $\epsilon_0$ as $\epsilon^2+2\epsilon$, such a lower bound $OPT_L=f^*_{\epsilon_0}(S_{L})$ satisfies $f^*_{\epsilon_0}(S_{L})\leq (1+\epsilon^2+2\epsilon)  f^*(S_{L})\leq (1+\epsilon)^2  f^*(S_{opt})$, and combining Lemma \ref{lemma: cost-R}, the expected running time of generating such an estimate is $O\Big( \frac{n\cdot m(1+\max_1 ) \cdot \ln N}{\epsilon^2 \cdot  OPT_L}\Big)$ for any small $\epsilon$. Since we would prefer a tight lower bound so that the running time can be minimized, we could select $S_L$ through intuitive heuristics. For example, selecting the our-neighbors of the nodes in $S_r$.

\subsection{The HMP Algorithm}
Now we are ready to present our algorithm. The algorithm is formally given in Alg. \ref{alg: HS}, termed as Hybrid-sampling based Misinformation Prevention (HMP) algorithm. As shown in Alg. \ref{alg: HS}, we first compute a lower bound $OPT_L$ of $(1+\epsilon)^2\cdot f^*(S_{\opt})$ according to Sec. \ref{subsec: lower}, then solve Eq. (\ref{eq: optimal}) to find the optimal setting of $\epsilon_{11}$ and $\epsilon_{12}$. Finally, we use the framework Alg. \ref{alg: framework} to obtain the solution $S^*$. 

\begin{algorithm}[t]
\caption{HMP Algorithm}
\label{alg: HS}
\begin{algorithmic}[1]
\State \textbf{Input:} $G$, $\epsilon$, $N$, $S_r$ and $k$;
\State Select any subset $S_L$ with $|S_L|=k$ and compute $OPT_L$ by Lemma \ref{lemma: optimal_monte} with $\epsilon_0=\epsilon^2+2\epsilon$;
\State Compute $\epsilon_{11}$ and $\epsilon_{12}$ by solving Eq. (\ref{eq: optimal}).
\State $l \leftarrow \max(l_1, l_2)/OPT_L$
\State $S^* \leftarrow$ Alg. \ref{alg: framework}$(G, S_r, k, l)$;
\State \textbf{Return} $S^*$;
\end{algorithmic}
\end{algorithm}

\textbf{Performance Guarantee.} The HMP algorithm has the following performance bounds.

\begin{theorem}
\label{theorem: bound}
With probability at least $1-3/N$, the HMP algorithm returns a solution $S^*$ such that \[f^*(S^*) \geq (1-1/e-\epsilon ) f^*(S_{\opt}).\]
\end{theorem}
\begin{proof}
Since $OPT_L \leq (1+\epsilon)^2 f^*(S_{\opt})$ holds with probability at least $1-1/N$, the theorem follows immediately from Lemma \ref{lemma: l}. 

\end{proof}

\textbf{Running Time.} In Alg. \ref{alg: framework}, there are totally $l$ R-samples generated where $l=O(\frac{n \cdot (k\ln n+\ln N)}{\epsilon^2\cdot OPT_L})$. By Lemma \ref{lemma: cost-R}, the expected running time of line 2 in Alg. \ref{alg: framework} is 
\begin{equation}
\label{eq: time}
O\Big(\frac{m \cdot n(k\ln n+\ln N)(1+\max_1)}{\epsilon^2\cdot OPT_L}\Big).
\end{equation}
Furthermore, Alg. \ref{alg: greedy} can be implemented to run in time linear to the total size of its input \cite{vazirani2013approximation}. Therefore, line 5 in Alg. \ref{alg: HS} runs in the same as Eq. (\ref{eq: time}). The rest part of Alg. \ref{alg: HS} is dominated by line 5 so the whole algorithm runs asymptotically in the same as line 5 does. 

For the one interested in a success probability of $1-n^{-l}$, the running time would be \[O\Big(\frac{m \cdot n\cdot \ln n \cdot  (k+l) (1+\max_1)}{\epsilon^2\cdot OPT_L}\Big).\]

\subsection{Comparing to the Existing Works.}
\label{subsec: compare}
There are three most related existing works, \cite{budak2011limiting, tong2017efficient, simpson2018scalable},  which aim at solving Problems \ref{problem: max_non} and \ref{problem: max_prevetion} by designing approximation algorithms. C. Budak \textit{et al.} \cite{budak2011limiting} first proposed the problem of limiting the spread of information and they considered Problem \ref{problem: max_prevetion}. In \cite{budak2011limiting}, it is shown that Problem \ref{problem: max_prevetion} is monotone and submodular, and the authors use Monte Carlo simulation straightforwardly (i.e., the forward sampling method) to overcome the difficulty in computing the objective function. G. Tong \textit{et al.} \cite{tong2017efficient} applied the reverse sampling technique invented by \cite{borgs2014maximizing, tang2014influence, tang2015influence} to solve Problem \ref{problem: max_non}. Very recently, M. Simpson \textit{et al.} \cite{simpson2018scalable} adopted the same technique to solve Problem \ref{problem: max_prevetion}. In the work \cite{tong2017efficient} and \cite{simpson2018scalable}, the authors utilized the same sampling method designed based on the uniform reverse sampling framework. As mentioned earlier, selecting a node uniformly is not efficient for the MN or MP problem, because it cannot attach extra importance to the nodes that are prone to be $C_r$-active. Furthermore, as discussed in \cite{simpson2018scalable}, for the MP problem, the sample obtained by the uniform reverse sampling method can be an empty set. When an empty set is sampled, it corresponds to the case that the selected node will not be affected by the misinformation and we do not need to select any positive seed node to protect that node. As shown in Alg. \ref{alg: greedy}, the empty set cannot provide any information for selecting effective nodes. On the other hand, in the hybrid sampling method, every sample we obtain by Alg. \ref{alg: reverse_v} is guaranteed to be non-empty, and therefore our sampling process is more effective in generating the samples that can help us find high-quality positive seed nodes. We have observed the same in our experiments.

\section{Experiments}
\label{sec: exp}
We now the present the experiments done for evaluating the proposed solution by comparing it to the state-of-the-art solutions, RBR \cite{tong2017efficient} and RPS \cite{simpson2018scalable}. 

\begin{figure*}[t]
\centering
\subfloat[Wiki under weighted-cascade model]{\label{fig: wiki_wc}\includegraphics[width=0.30\textwidth]{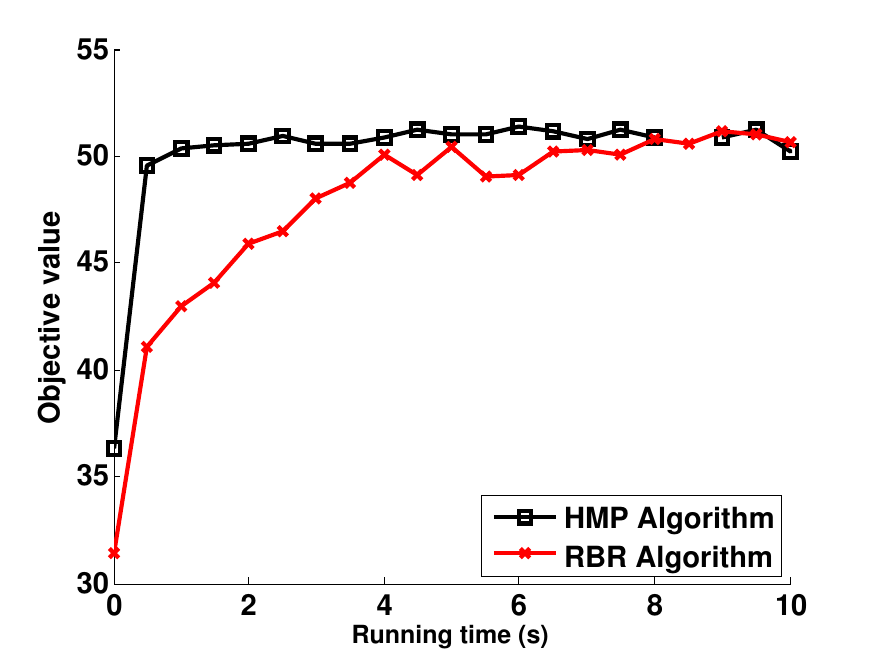}} \hspace{1mm}  
\subfloat[HepPh under weighted-cascade model]{\label{fig: hepph_wc}\includegraphics[width=0.30\textwidth]{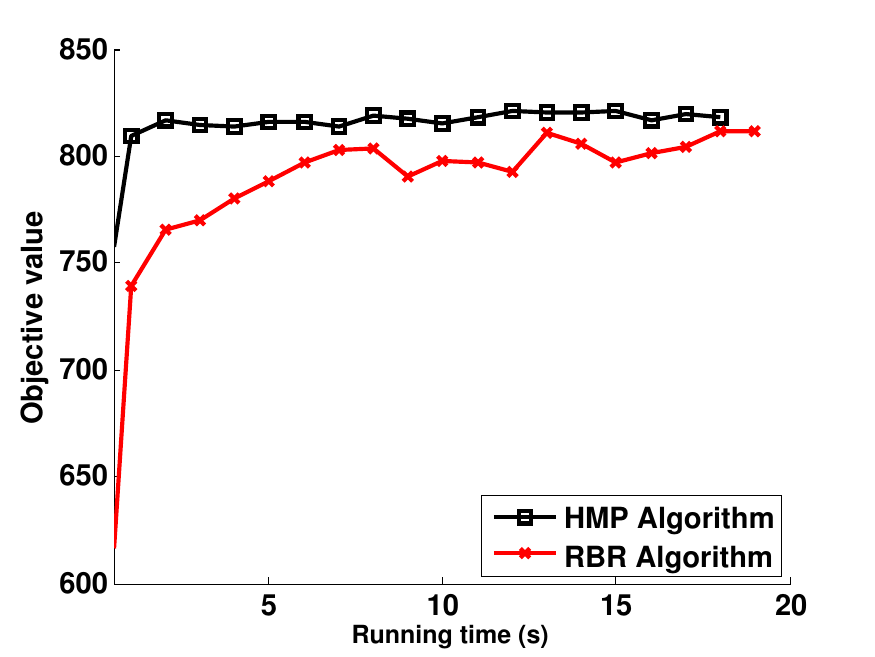}} \hspace{1mm}  
\subfloat[HepPh under uniform 0.1]{\label{fig: hepph_ic}\includegraphics[width=0.30\textwidth]{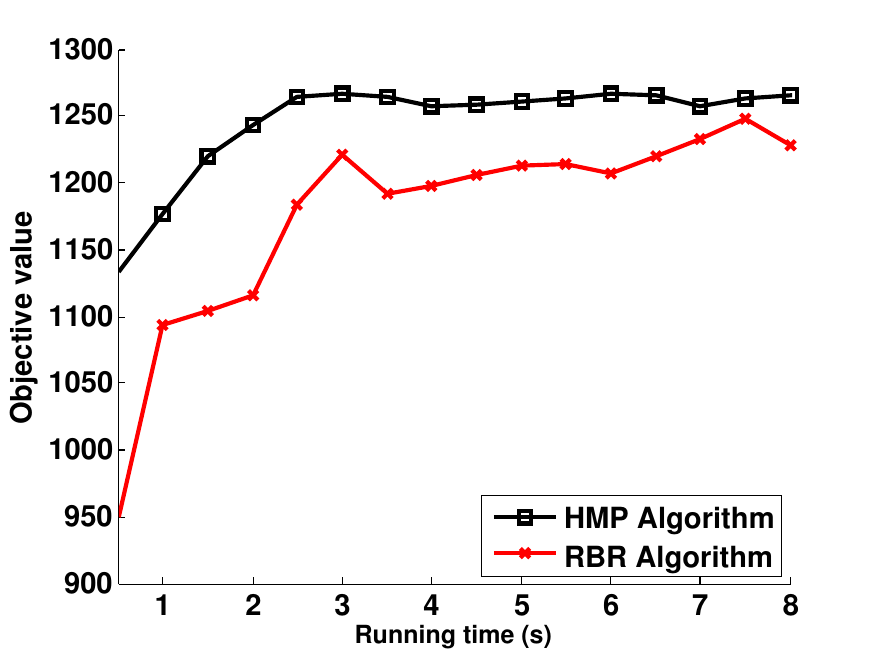}} \hspace{1mm}  
\vspace{-2mm}
\subfloat[Epinions under weighted-cascade model]{\label{fig: epinions_wc}\includegraphics[width=0.30\textwidth]{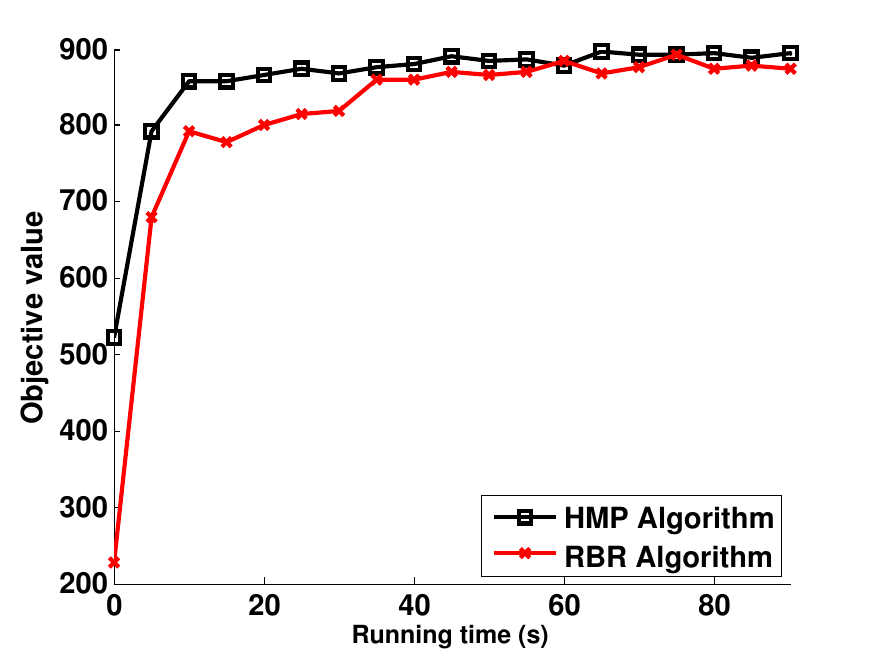}} \hspace{1mm}  
\subfloat[Epinions under uniform 0.01]{\label{fig: epinions_ic}\includegraphics[width=0.30\textwidth]{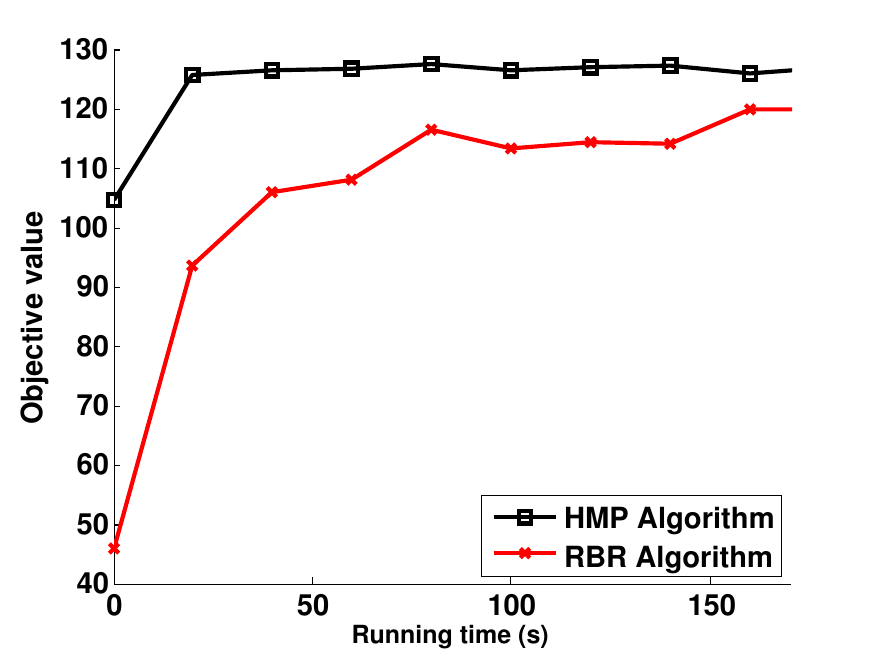}} \hspace{1mm}  
\subfloat[DBLP under weighted-cascade model]{\label{fig: dblp_wc}\includegraphics[width=0.30\textwidth]{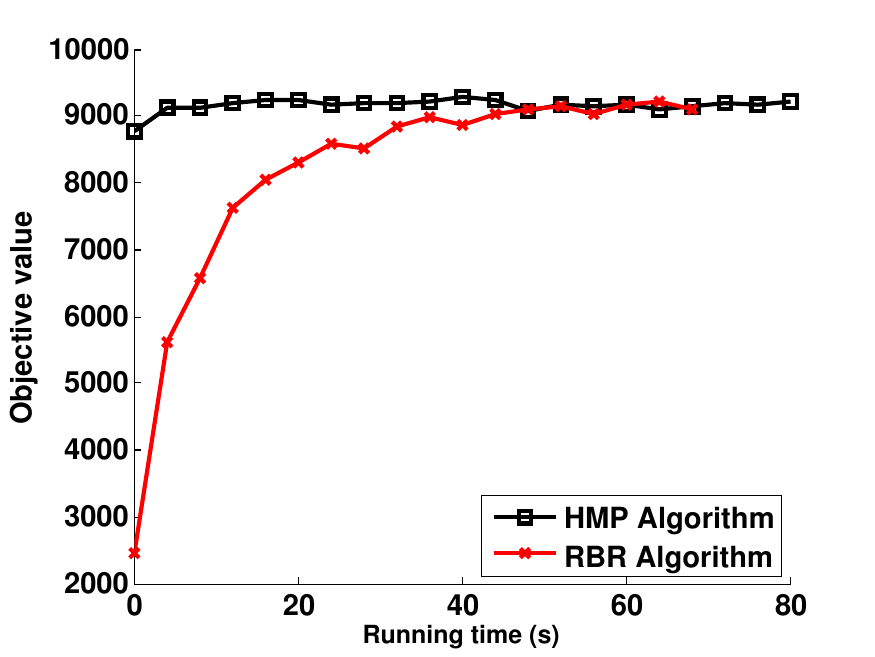}} \hspace{0mm}  
\vspace{-2mm}
\subfloat[DBLP under uniform 0.1]{\label{fig: dblp_ic}\includegraphics[width=0.30\textwidth]{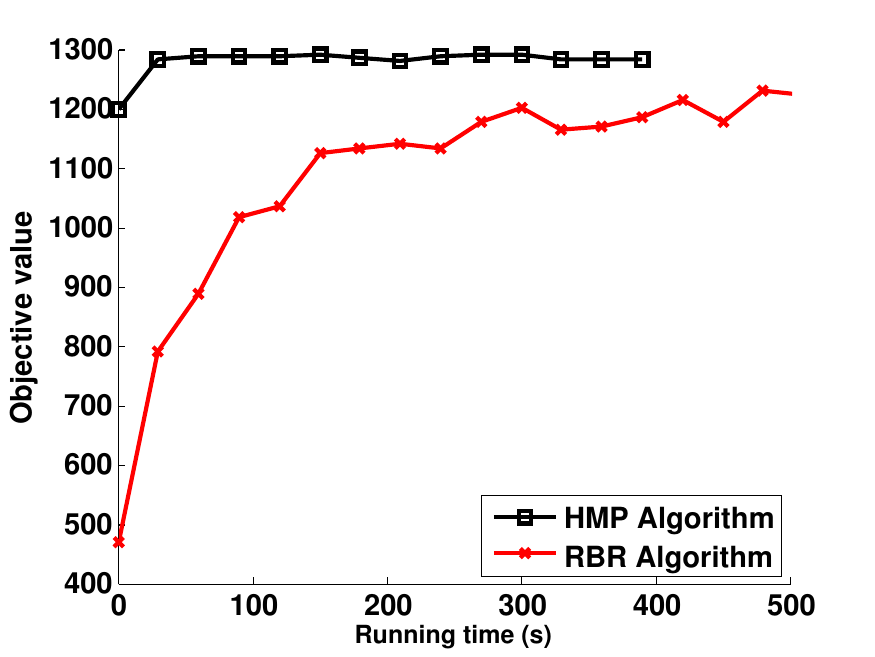}} \hspace{0mm}  
\subfloat[Youtube under weighted-cascade model]{\label{fig: youtube_wc}\includegraphics[width=0.30\textwidth]{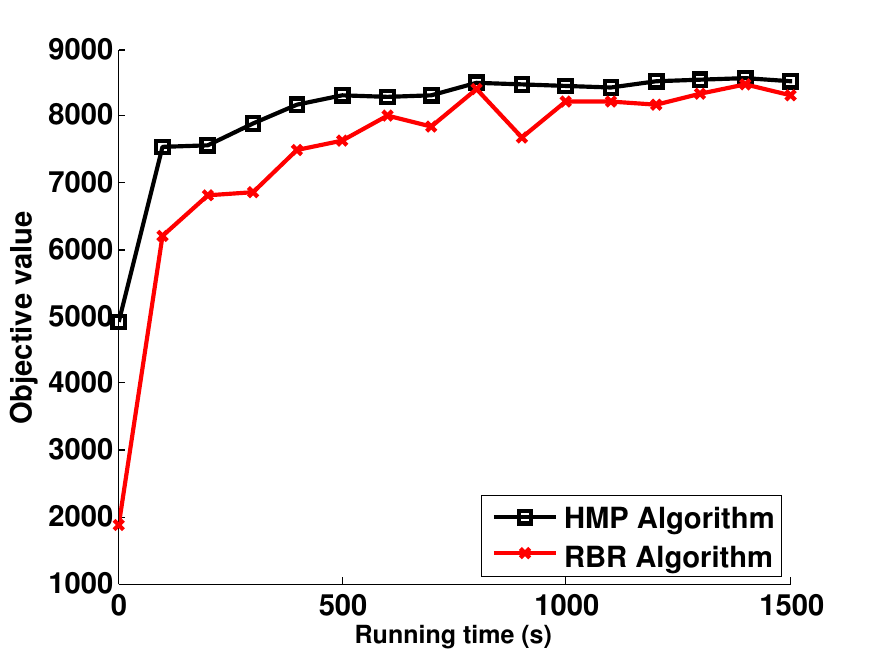}} \hspace{0mm}  
\subfloat[Youtube under uniform 0.01]{\label{fig: youtube_ic}\includegraphics[width=0.30\textwidth]{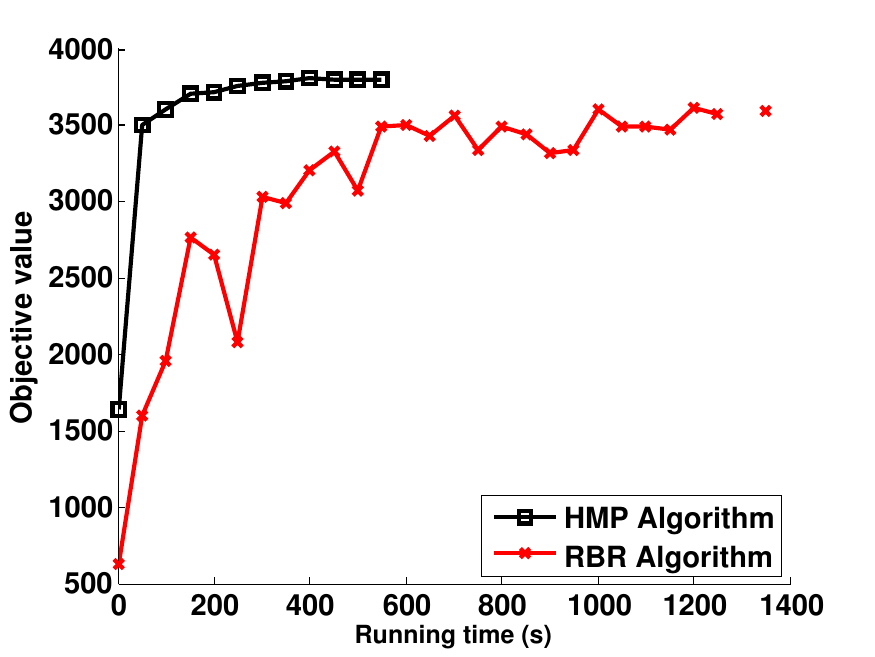}} \hspace{0mm}  
\vspace{-0mm}
\caption{Experimental Result.} 
\label{fig: exp}
\vspace{-3mm}
\end{figure*}

\subsection{Experimental Setting}

In most of the existing works, the performance of an algorithm is measured by either the running time or the quality of the produced solution. However, for two algorithms namely A and B, it happens that A produces a better solution than that produced by B, but meanwhile, B runs much faster than A does. Under such a scenario, one can hardly decide that which algorithm is better. In order to fairly compare the algorithms, we define the performance of an algorithm as the quality of the produced solution under a specific time constraint. 

Note that all the three algorithms, HMP, RBR and RPS, follow the framework shown in Alg. \ref{alg: framework}, and the only difference is that HMP utilizes a different method to generate samples. The three algorithms all have the following two properties: (1) the more samples we use, the better solution we obtain; (2) the quality of the solution finally converges to an absolute value upper bounded by the optimal value. Therefore, by controlling the number of the used samples, there is a natural trade-off between the running time and quality of the solution. In this way, we are able to see that with the same running time which algorithm is the best. As discussed in Sec. \ref{subsec: compare}, the RBR algorithm and the RPS algorithm essentially use the same method to generate samples, so we only compare HMP with RBR.

\begin{table}[t]
\renewcommand{\arraystretch}{1.3}
\centering
{\begin{tabular}{ C{1cm} || C{0.8cm}| C{0.8cm} |p{4cm}}
\textbf{Dataset} 	& \textbf{Node} 		& \textbf{Edge} 		& \textbf{Description} \\   
\hline
\hline
Wiki 	& 7K		& 103K 		& Wikipedia who-votes-on-whom network \\
\hline
HepPh 	& 34K		& 421K 		& Arxiv High Energy Physics paper citation network \\
\hline
Epinions 	& 75K		& 508K 		& Who-trusts-whom network of Epinions.com \\
\hline
DBLP 	& 317K		& 1.04M 		& DBLP collaboration network \\
\hline
Youtube 	& 1.1M		& 6.0M 		& Youtube online social network \\

\end{tabular}}
\caption{{Datasets}}
\label{table: datasets}
\end{table}

\textbf{Experimental Setting.} We adopt five datasets, from small to large, Wiki, HepPh, Epinions, DBLP, and Youtube, borrowed from SNAP \cite{snapnets}. A brief description of each dataset is given in Table. \ref{table: datasets}. We consider three types of propagation probability: (1) the uniform setting where $\Pr[e]=0.1$; (2) the uniform setting where $\Pr[e]=0.01$; (3) the weighted-cascade model where $\Pr[(u,v)]=1/d_v$ for each $(u,v) \in E$ and $d_v$ is the in-degree of $v$. These settings have been widely adopted in prior work. We set $k=10$ and $|S_r|=15$ in our experiments, unless otherwise specified. The nodes in $S_r$ are selected from the nodes with the highest individual influence. For each algorithm and dataset, we continuously increase the number of the used samples and record the produced positive the seed set, until the quality of the solution tends to converge. Each solution is finally evaluated by 10,000 Monte Carlo simulations. 

\subsection{Experimental Results}  
The experimental results are shown in Fig. \ref{fig: exp}. Each figure gives two curves plotting the results of HMP and RBR respectively. The minimum unit of the running time is 0.1 second, and it is shown as 0 when the running time is less than 0.1 second.

\textbf{Major observations.} The most important observation is that HMP consistently outperforms RBR under all considered datasets and propagation probability settings. Under the same time constraint, the solution of HMP is more effective than that of RBR with respect to maximizing $f^*(S)$. For example, in Fig. \ref{fig: epinions_wc}, when the running time is required to be less than one second, HMP is able to protect more than 500 users while RBR can only protect 220 users. Furthermore, HMP takes much less time to reach the maximal performance. As an extreme example, in Fig. \ref{fig: dblp_wc}, HMP converges after 5 seconds while RBR converges after 40 seconds. In other words, a few samples will be sufficient for HMP to achieve a high performance. Finally, the curves of HMP is smoother than that of RBR, which means HMP is more robust than RBR. Such a scenario suggests that the variance of the solution produced by RBR is typically higher than that of the solution given by HMP. 

\textbf{Minor observations.} According to the figures, HMP usually has a breakpoint after which the quality of the solution remains unchanged. For example, (0.25, 49) in Fig. \ref{fig: wiki_wc} and (50, 3500) in Fig. \ref{fig: youtube_ic}. On the other hand, curves of RBR tend to increase gradually and there is no significant breakpoint. It would be great if we can identify the numbers of the samples corresponding to such breakpoints, because they are the minimum number to reach the maximal performance. Unfortunately, we find that the theoretical bound given in Sec. \ref{sec: algortihm} usually far exceeds the minimum requirement of the breakpoints. One possible reason is that the theoretical bounds are dealing with the worst-case scenario so it is too pessimistic for the average case. Another reason is that we assume the variance is $n^2$ in the Chernoff bound, which leads to an overestimate of the number of the required samples. The variance of the distribution over the diffusion outcomes can significantly help us to obtain an accurate estimate. Currently, we are not aware of any existing work which has considered this issue.

\begin{remark}
\label{remark: ratio}
It should be noted that when the objective function stops increasing, its maximal performance has been almost reached. When the maximal performance is reached, the solution produced at that moment is no worse than a $(1-1/e)$-approximation as $1-1/e$ is the worst-case ratio. 
\end{remark}

\section{Conclusion and Future Work}
\label{sec: con}
In this paper, we present a hybrid sampling method which is designed particularly for the misinformation prevention problem. We show that the new sampling method can be used to design an approximation algorithm which outperforms the state-of-the-art solutions. 

We propose some future works which we believe are interesting. It has been shown that there exist $O(k \cdot m \cdot  \log n/\epsilon^2)$ algorithms for the MN and IM problem. Though the MP problem and the MN problem merely differ by a constant, the best algorithm for the MP problem is $O(k \cdot m \cdot  n \cdot \log n/\epsilon^2)$ (assuming $OPT_L\geq (1+\max_1)$), and it is open that if we can find a $O(k \cdot m \cdot \log n/\epsilon^2)$ algorithm for the MP problem. In addition, the existing lower bounds on the number of the samples used to reach the breakpoints observed in the experiments are still pessimistic. We are looking for more practical methods to help us find such breakpoints so that we can achieve the best performance with the minimum running time.

\section*{Acknowledgment}
We thank the reviewer for pointing out the issue discussed in Remark \ref{remark: ratio}. This work is supported in part by the start-up grant from the University of Delaware. 

\bibliographystyle{IEEEtran}
\bibliography{sigproc}

\end{document}